\documentclass{article}
\usepackage[margin=1in]{geometry}
\usepackage{amsmath,amssymb,amsthm}
\usepackage{graphicx}
\usepackage{cite}
\usepackage{hyperref}
\usepackage{microtype}

\newtheorem{theorem}{Theorem}

\newtheorem{lemma}{Lemma}

\newtheorem{definition}{Definition}

\newtheorem{corollary}{Corollary}

\newtheorem{observation}{Observation}

\title{Simple Recognition of Halin Graphs and Their Generalizations}
\author{David Eppstein\thanks{Computer Science Department, University of California, Irvine; Irvine, CA, USA. This material is based upon work supported by the National Science Foundation under Grant CCF-1228639 and by the Office of Naval Research under Grant No. N00014-08-1-1015.}}

\date{ }

\begin{document}
\maketitle

\begin{abstract}
We describe and implement two local reduction rules that can be used to recognize Halin graphs in linear time, avoiding the complicated planarity testing step of previous linear time Halin graph recognition algorithms. The same two rules can be used as the basis for linear-time algorithms for other algorithmic problems on Halin graphs, including decomposing these graphs into a tree and a cycle, finding a Hamiltonian cycle, or constructing a planar embedding.
These reduction rules can also be used to recognize a broader class of polyhedral graphs. These graphs, which we call the D3-reducible graphs, are the dual graphs of the polyhedra formed by gluing pyramids together on their triangular faces; their treewidth is bounded, and they necessarily have Lombardi drawings. 
\end{abstract}

\section{Introduction}

\emph{Halin graphs} are the graphs that can be formed from a tree with no degree-two vertices, embedded in the plane, by adding a cycle of edges connecting the leaf vertices of the tree in the cyclic order given by the embedding~\cite{Hal-CMA-71}. They are necessarily 3-connected  and planar, with a unique planar embedding up to the choice of the outer face of the embedding; we will adopt the convention that this outer face is the leaf cycle.
Unlike planar graphs more generally, Halin graphs have bounded treewidth (at most three), allowing problems such as the maximum independent set problem which are NP-hard on planar graphs to be solved in polynomial time on Halin graphs~\cite{Bod-ICALP-88}.

Two algorithms for recognizing Halin graphs in linear time are known. Sys{\l}o and Proskurowski~\cite{SysPro-GT-83} showed that a graph with $n$ vertices and $m$ edges is Halin if and only if it is planar and 3-connected, and has a face with exactly $m-n+1$ vertices and edges. All of these conditions can be checked in linear time. Fomin and Thilikos~\cite{FomThi-JDA-06} instead observed that in a Halin graph, the outer face has at least $n/2+1$ vertices, and that any planar graph can have at most four such faces. They proposed a recognition algorithm that constructs an (arbitrary) planar embedding, and tests for each large face whether its vertices all have degree three and whether removing the edges of the face from the graph leaves a tree with no degree-two vertices. Because there are only a constant number of faces to test, all steps of this algorithm can be performed in linear time. However, both of these algorithms use planarity testing, a problem whose many known linear-time algorithms~\cite{HopTar-JACM-74,BooLue-JCSS-76,ChiNisAbe-JCSS-85,ShiHsu-TCS-99,BoyMyr-JGAA-04,FraOssRos-IJFCS-06,Sch-MFCS-13} are complex and hard to implement. Linear-time 3-connectivity testing, also, has complex algorithms that have proven treacherous to implementors~\cite{HopTar-SJC-73,GutMut-GD-00}.

It would also be possible to base a linear time recognition algorithm for Halin graphs on Courcelle's theorem, which states that the monadic second-order logic of graphs has efficient decision algorithms for graphs of bounded treewidth~\cite{Cou-IC-90}. The existence of a decomposition of the edges of a given graph into a tree and a cycle through the leaves of the tree is straightforward to express in second-order logic. Expressing the correct ordering of the cycle with respect to the planar embedding of the tree is not as straightforward, but can be expressed logically as the statement that every subtree of the tree contacts a contiguous subpath of the cycle. Thus, to test whether a graph is Halin, one can construct a width-three tree-decomposition~\cite{Bod-SJC-96} and then check whether these logical expressions are valid for the decomposition. Such methods are again unlikely to lead to simple, practical, and implementable algorithms, because of the high constant factors resulting from the use of Courcelle's theorem. However, they could be used to recognize Halin graphs in logarithmic space~\cite{ElbJakTan-FOCS-10}.

An alternative approach that has proven successful for many other computational problems on graphs of bounded treewidth involves the notion of a \emph{reduction algorithm}, an algorithm that gradually shrinks the size of the input graph by applying \emph{reduction rules} based on local structures within a given graph~\cite{Flu-PhD-97}. If the reduction rules are chosen to be safe (preserving the property to be tested), complete (applicable to any large enough graph with the property), and terminating (always reducing some appropriate size function of the graphs they operate on), then the graphs with the property can be recognized by repeatedly applying reductions until no more can be found to apply, and then testing whether the remaining smaller graph belongs to a finite set of base cases. Such algorithms have been found for many specific graph classes~\cite{Duf-JMAA-65,ValTarLaw-SJC-82,ArnPro-SJADM-86,BodThi-Algs-99} and more generally are known to exist for all graph classes of bounded treewidth that can be recognized using Courcelle's theorem~\cite{ArnCouPro-JACM-93,Flu-PhD-97}. Thus, in particular, a reduction algorithm exists for recognizing Halin graphs. However, although (unsafe) reduction rules for Halin graphs have been used to solve the Steiner tree and edge-constrained Hamiltonian cycle problems in these graphs~\cite{SkoSys-ZM-87,Win-DAM-87}, to our knowledge, no explicit reduction algorithm for recognizing Halin graphs has been described.

Motivated by these considerations, we describe in this paper two simple reduction rules for Halin graphs that are safe, complete, and terminating, and that can be used to recognize Halin graphs in linear time. Our rules involve augmenting the vertices of the graph with a single additional bit of information, which we can interpret as a color of a vertex, black or white. We use these vertex colors to allow or disallow certain reductions, and then recolor certain vertices after each reduction.
The same two rules can be used as the basis for linear-time algorithms for other algorithmic problems on Halin graphs, including decomposing these graphs into a tree and a cycle, finding a Hamiltonian cycle, or constructing a planar embedding.

It is natural to consider the graphs obtained by simplifying these rules even further by leaving the graph vertices uncolored and allowing all reductions. We call the class of graphs that are recognized by the uncolored version of our two reduction rules the \emph{D3-reducible graphs} because the preconditions for both reduction rules involve triples of degree three vertices. As we show, the D3-reducible graphs are a generalization of the Halin graphs that, like the Halin graphs themselves, are automatically planar and 3-vertex-connected. Thus, by Steinitz's theorem~\cite{Ste-EMW-22}, they are the graphs of polyhedra, and we characterize the D3-reducible graphs geometrically as the dual graphs of the polyhedra that can be constructed by gluing together pyramids on their triangular faces. Additionally, we show that the D3-reducible graphs have treewidth at most four, and that they necessarily have planar \emph{Lombardi drawings}, drawings in which the edges are represented by circular arcs that meet at equal angles at each vertex. Planar Lombardi drawings were previously known to exist for Halin graphs and for planar graphs of maximum degree three~\cite{DunEppGoo-JGAA-12,Epp-DCG-14}, but beyond these classes their existence is somewhat mysterious; we do not even know whether they exist for all outerplanar graphs~\cite{LofNol-GD-12}.

\section{D3 reductions}

If $T$ is a tree with four or more vertices, none of degree two, it can be reduced to $K_{1,3}$ by reduction steps that either remove the two leaf children from a vertex of degree three or remove a leaf from a vertex of degree greater than three. Taking into account the cycle edges added to such a tree to form a Halin graph gives us the following two reduction rules:

\begin{description}
\item[D3a.] Let $p$, $q$, and $r$ be three degree-three vertices that induce a triangle in the given graph $G$, and whose neighbors outside the set $\{p,q,r\}$ are all distinct. Replace these three vertices by a single vertex with the same three outside neighbors. (\autoref{fig:reductions}, left.)
\item[D3b.] Let $p$, $q$, and $r$ be three degree-three vertices that induce a path, with $q$ as the middle vertex, and suppose additionally that there is a single vertex $s$ adjacent to all three of $p$, $q$, and $r$.
Delete $q$ from the graph and replace it by a new edge from $p$ to $r$. (\autoref{fig:reductions}, right.) We refer to $s$ as the \emph{apex} of the reduction and $q$ as the \emph{middle vertex} of the reduction.
\end{description}

\begin{figure}[t]
\centering
\includegraphics[scale=0.35]{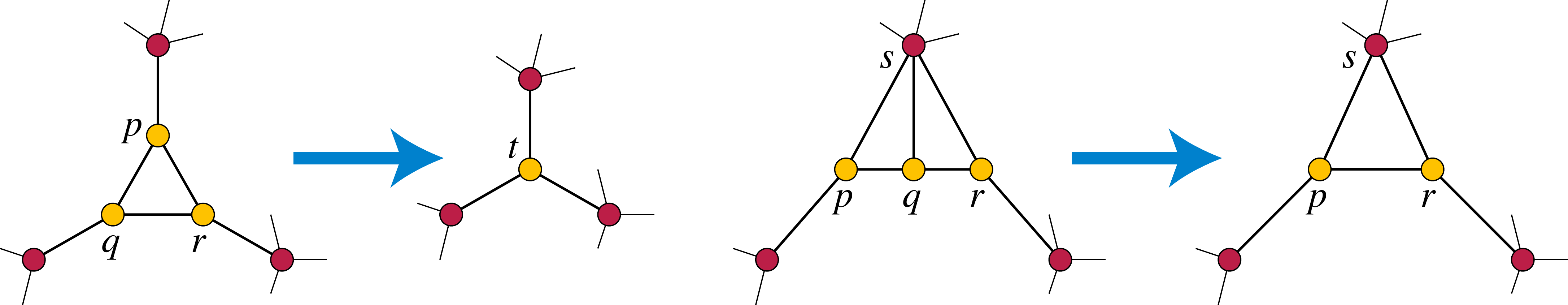}
\caption{The two D3 reductions. Left: three degree-three vertices $p$, $q$, and $r$ form a triangle with three distinct neighbors, and are collapsed to a single vertex $t$. Right: three degree-three vertices $p$, $q$, and $r$ form a path with one shared neighbor $s$, and are contracted to a two-vertex path.}
\label{fig:reductions}
\end{figure}

We collectively refer to rules D3a and D3b as the D3 reductions.

\begin{definition}
We define a \emph{D3-reducible graph} to be a graph that can be reduced to the four-vertex complete graph $K_4$ by a sequence of D3 reductions. We define a graph to be \emph{irreducible} if no additional D3 reductions can be applied to it.
\end{definition}

For instance, $K_4$ is irreducible, because all triples of its degree-three vertices induce triangles but do not have three distinct neighbors outside of each of these triangles.
As we now show, it will not be necessary to search for a special reduction sequence that leads to $K_4$ in order to recognize these graphs: all reduction sequences lead to isomorphic graphs.

\begin{lemma}
\label{lem:interchangeable}
Let $G$ be any graph, and $X$ and $Y$ be two D3 reductions that are both applicable in $G$.
Then either $X$ and $Y$ may both be applied independently (in either order) or the result of performing $X$ is isomorphic to the result of performing $Y$.
\end{lemma}

\begin{figure}[b]
\centering
\includegraphics[scale=0.35]{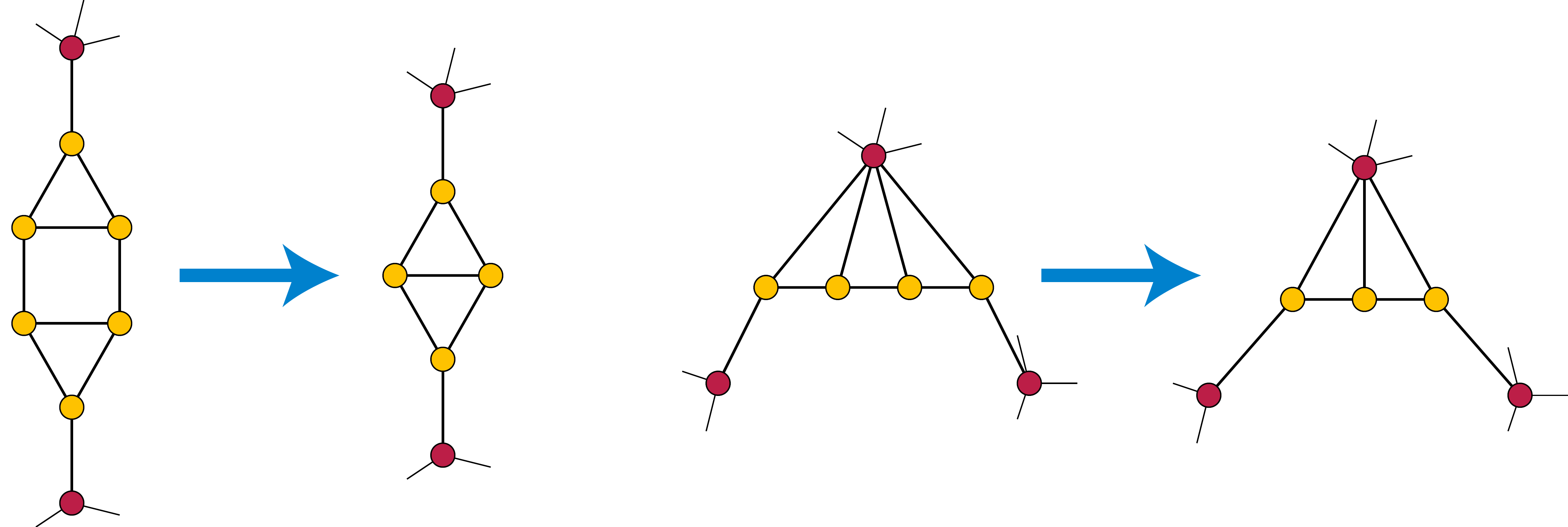}
\caption{Non-independent pairs of D3 reductions. Left: D3a reductions collapsing either of the two yellow triangles give isomorphic results. Right: D3b reductions shortening either of the two overlapping paths of three yellow vertices give isomorphic results.}
\label{fig:confluent}
\end{figure}

\begin{proof}
If $pqr$ is a triangle of degree-three vertices, and a D3 reduction that is not a D3a reduction of $pqr$ is performed, then $pqr$ remains a triangle of degree-three vertices.
And if $pqr$ is an induced path with a shared neighbor $s$, and a D3 reduction that is not a D3b reduction of $pqr$ is performed, then $q$ remains the middle vertex of an induced path with a shared neighbor~$s$. So the only way one reduction $X$ could prevent the future performance of another reduction $Y$ is by changing the neighbors of a middle vertex $q$ of a D3b reduction or causing the outside neighbors of a triangle to become non-distinct.
For this to happen, we have the following cases:
\begin{itemize}
\item If $X$ and $Y$ are both D3a reductions, then their two triangles of degree-three vertices are connected to each other by two or three edges. If they are connected by three edges, then the result of performing either reduction is $K_4$. If they are connected by only two edges, then after either of the two reductions the result is a graph in which the two non-adjacent vertices of the two triangles are linked by a pair of triangles that share an edge (\autoref{fig:confluent}, left).
\item If $X$ and $Y$ are both D3b reductions, then they can only affect each other if their two paths share an edge, so that they are both part of a path of four degree-three vertices that all are adjacent to the same apex. In this case the results of performing either reduction are a graph in which this four-vertex path has been replaced by a three-vertex path (\autoref{fig:confluent}, right).
\item If one of $X$ and $Y$ is a D3a reduction and the other is a D3b reduction, then neither reduction can affect the other one: a D3a reduction can't change the neighbors of the middle vertex of a path, and a D3b reduction can't make previously-distinct vertices become the same as each other.\qedhere
\end{itemize}
\end{proof}

\begin{lemma}
\label{lem:confluence}
In a 3-vertex-connected graph, all maximal sequences of D3 reductions lead to isomorphic irreducible graphs.
\end{lemma}

\begin{proof}
Let $\Sigma_1$ and $\Sigma_2$ be two maximal sequences of reductions starting from the same graph $G$, and let $X$ be the first reduction in $\Sigma_1$. We will prove that there exists a reduction sequence $\Sigma_3$ such that $\Sigma_3$ begins with $X$, and such that $\Sigma_2$ and $\Sigma_3$ transform $G$ into isomorphic graphs. The claimed result will then follow by induction on the length of the two sequences,
by applying the induction hypothesis to the graph obtained from $G$ by reduction $X$. As a base case, when either $\Sigma_1$ or $\Sigma_2$ is empty, then $G$ is already irreducible, and there can only be one maximal sequence of reductions (the empty sequence).

To show that every sequence $\Sigma_2$ has an equivalent sequence $\Sigma_3$ beginning with $X$, we again use induction on the length of $\Sigma_2$. $\Sigma_2$ cannot be empty, for operation $X$ can be applied to $G$ and the result of applying $\Sigma_2$ should be an irreducible graph; therefore, we can define $Y$ to be the first operation in $\Sigma_2$. We have three cases:
\begin{itemize}
\item If $X=Y$ then $\Sigma_2$ already begins with $X$ and the result follows.
\item If $X$ and $Y$ both give isomorphic graphs then we can create the desired sequence $\Sigma_3$ by replacing $Y$ by $X$ and again the result follows.
\item In the remaining case, by \autoref{lem:interchangeable}, $X$ and $Y$ may be applied independently. Let $H$ be the graph obtained from $G$ by reduction $Y$, and let $\Sigma_4$ be the reduction sequence on~$H$ obtained from $\Sigma_2$ by removing the first reduction~$Y$.
Then $X$ may be applied to $H$, and by induction there exists a reduction sequence $\Sigma_5$ on~$H$ that begins with $X$ and has the same effect as $\Sigma_4$. The desired reduction sequence $\Sigma_3$ may be obtained by applying reductions $X$ and $Y$ followed by the remaining reductions (after~$X$) in~$\Sigma_5$.\qedhere
\end{itemize}
\end{proof}

In the language of rewriting systems, \autoref{lem:confluence} means that D3 reductions are \emph{confluent}, or have the \emph{Church--Rosser property}. This allows us to apply them greedily without worrying about the ordering of the reductions.

\begin{theorem}
We can recognize D3-reducible graphs in linear time.
\end{theorem}

\begin{proof}
We maintain an adjacency list representation of the graph $G$ after a sequence of reductions, allowing edge insertions and removals, degree tests, finding the endpoints of an edge, and finding the neighbors of a bounded-degree vertex in constant time per operation. We also maintain a collection $C$ of vertices or former vertices of the graph that is guaranteed to contain at least one of the three degree-three vertices of each possible D3a reduction and the middle vertex of each possible D3b reduction. Our algorithm performs the following steps:
\begin{enumerate}
\item Initialize $C$ to the set of all degree-three vertices in $G$.
\item While $C$ is non-empty:
\begin{enumerate}
\item Select and remove an arbitrary vertex $v$ from $C$.
\item If $C$ and two of its neighbors form the triangle of a D3a reduction, perform that reduction; add the new vertex formed by the reduction and all its degree-three neighbors to $C$.
\item Otherwise, if $C$ and two of its neighbors form the path of a D3b reduction, perform that reduction; if this causes the apex of the reduction to have degree three, add it to $C$.
\end{enumerate}
\item Test whether the resulting graph is $K_4$.
\end{enumerate}
$C$ initially contains $O(n)$ vertices. Each successful reduction adds $O(1)$ vertices to it, so the total number of vertices ever added to $C$ is linear. The time for the algorithm is $O(n)$ for the initialization and final testing stages, and $O(1)$ per vertex in $C$ for the inner loop, giving $O(n)$ in total. The correctness of the algorithm follows from \autoref{lem:confluence}.
\end{proof}

We remark that the graphs that can be obtained using only D3a reductions are the dual graphs of planar 3-trees, and the graphs that can be obtained using only D3b reductions are the wheel graphs.

\section{Shared properties with Halin graphs}

As we now show, D3-reducible graphs have many of the same properties that Halin graphs are known to have. We will use this fact to simplify some algorithmic computations on Halin graphs, such as the search for Hamiltonian cycles, by generalizing these computations to D3-reducible graphs.

\begin{theorem}
Every D3-reducible graph is 3-vertex-connected.
\end{theorem}

\begin{proof}
We use induction on the number of D3 reductions in an (arbitrarily chosen) sequence of reductions that takes the given graph $G$ to $K_4$. As a base case, $K_4$ itself is connected. Otherwise, let $H$ be the graph formed from $G$ by the first reduction in the sequence. By induction, every two pairs of vertices in $H$ have three vertex-disjoint paths connecting them, and we must show that the same is true in $G$. We divide into cases:
\begin{itemize}
\item For pairs of vertices $x,y$ outside the set of three vertices $p$, $q$, and $r$ defining the reduction,
the three paths in $H$ connecting $x$ to $y$ may be straightforwardly modified to give three paths in $G$, replacing a path through the collapsed vertex of a D3a reduction by a path through two vertices, and a path through edge $pr$ of a D3b reduction by a path through edges $pq$ and $qr$.
\item For pairs of vertices one of which is $p$, $q$, or $r$, the required three vertex-disjoint paths connecting the pair of vertices can be found by replacing $p$, $q$, or $r$ by one of the corresponding vertices in $H$, finding three vertex-disjoint paths in $H$, and again making straightforward modifications to find three vertex-disjoint paths in the original graph.
\item In the remaining case, we are given two vertices of $p$, $q$, and $r$, and must find three vertex-disjoint paths connecting them in $G$. First, suppose that $H$ is obtained by a D3a reduction; by symmetry, we may assume that we are finding paths connecting $p$ and $q$. In this case, two such paths exist within the triangle defining the D3a reduction; the third path can be found as one of the three paths connecting the outside neighbors of $p$ and $q$. Second, suppose that $H$ is obtained from a D3b reduction and that we are connecting vertices $p$ and $r$. Again, two of the required paths from $p$ to $r$ exist: the induced path $pqr$ and the path $psr$ through the apex of the reduction. The third path can be found as one of the three paths connecting the outside neighbors of $p$ and $r$. Third and finally, suppose that we are connecting vertices $p$ and $q$ of a D3b reduction. In this case, we have paths $pq$ and $psq$ within the reduced part of the graph, and a third path through $r$ together with one of the three paths connecting the outside neighbors of $p$ and $r$.
\end{itemize}
Thus, for all pairs of vertices in $G$, there exist three vertex-disjoint paths, and the result follows.
\end{proof}

An alternative proof of this theorem can be given by applying a known reduction-based characterization of the 3-vertex-connected graphs: a graph is 3-connected if and only if it can be reduced to $K_4$ by the Barnette--Gr\"unbaum reduction rules~\cite[Thm.~1]{BarGru-MFGT-69}. These rules allow the removal of any edge from a graph and the suppression of any vertex of degree two created by this removal. Arbitrary edge removals can fail to reach $K_4$, but a sequence of reduction steps leading to $K_4$ can be found in polynomial time when it exists~\cite{Sch-Algo-12}. The same reduction rules have also been used as part of a linear-time planarity test~\cite{Sch-MFCS-13}. In our case, a D3a reduction can be seen as a Barnette--Gr\"unbaum reduction that removes a triangle edge, and a D3b reduction can be seen as a Barnette--Gr\"unbaum reduction that removes the edge between the middle vertex of a path and the apex of the reduction. Therefore, the reduction sequence for a D3-reducible graph also gives a Barnette--Gr\"unbaum reduction sequence taking the graph to~$K_4$.

For the next property of D3-reducible graphs, recall that, when a 3-vertex-connected graph is planar, its planar embedding is unique up to the choice of the outer face, and its faces are exactly the induced cycles for which the graph induced by the complementary set of vertices is connected~\cite{Bru-JCTB-04}.

\begin{theorem}
Every D3-reducible graph is planar, and every triangle in the graph is a face of its unique planar embedding.
\end{theorem}

\begin{proof}
We use induction on the number of D3 reductions; as a base case, $K_4$ clearly has the stated properties. For any other D3-reducible graph $G$, suppose that the graph has a D3 reduction $X$ leading to a smaller graph $H$; by the induction hypothesis, $H$ is planar with all triangles as faces. We have two cases:
\begin{itemize}
\item If $X$ is a D3a reduction, then $G$ may be obtained from $H$ by replacing a degree-three vertex $v$ by a triangle. A planar embedding of $G$ may be obtained from the embedding of $H$ by adding one new edge to each of the three faces that meet at $v$, and forming a new face triangle from the three new edges. The only new triangle created by this replacement is necessarily a face.
\item If $X$ is a D3b reduction, then $G$ may be obtained from $H$ by subdividing an edge $uv$ that belongs to a triangle $uvw$ and connecting the new subdivision vertex to the opposite apex $w$ of the triangle. An embedding of $G$ may be obtained in the same way, by splitting the face $uvw$ of the embedding of $H$ into two new triangular faces. The two new triangles formed from the subdivision are again faces of the subdivided embedding.\qedhere
\end{itemize}
\end{proof}

The proof of this result may be used to derive a linear-time algorithm to construct a planar embedding of a D3-reducible graph, more simply than using a general-purpose linear-time planar embedding algorithm, by reversing the sequence of reductions and maintaining a planar embedding for each step of the reversed reduction sequence.

\begin{theorem}
\label{thm:ham}
Every D3-reducible graph has a Hamiltonian cycle that can be found in linear time.
\end{theorem}

\begin{figure}[t]
\centering
\includegraphics[scale=0.35]{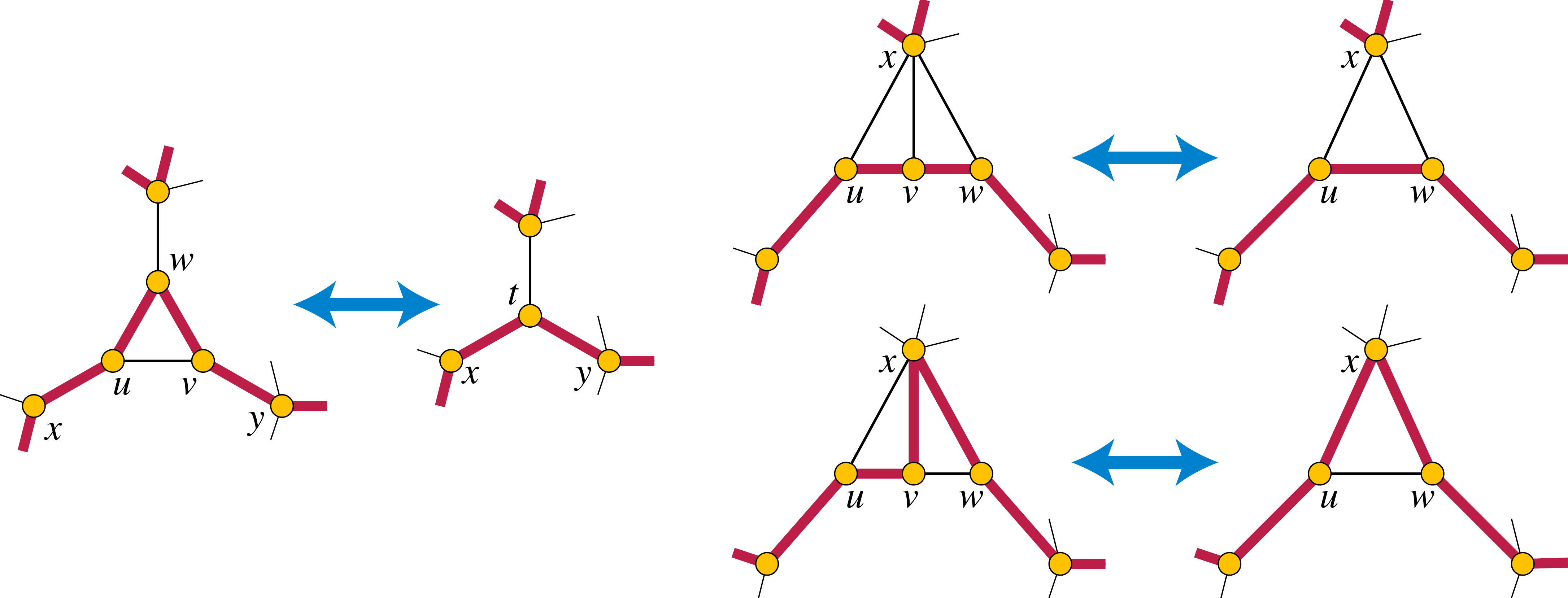}
\caption{Cases for \autoref{thm:ham}. In each case the rightward pointing arrow describes a D3 reduction from the given graph $G$ to a smaller graph $H$, and the leftward pointing arrow shows how to modify the Hamiltonian cycle $C$ in $H$ (thick red edges) to give a Hamiltonian cycle in~$G$.}
\label{fig:ham}
\end{figure}

\begin{proof}
We use induction on the number of D3 reductions.
As a base case, $K_4$ is Hamiltonian. For any other D3-reducible graph $G$, suppose that the graph has a D3 reduction $X$ leading to a smaller graph $H$; by the induction hypothesis, $H$ has a Hamiltonian cycle $C$. We have three cases:
\begin{itemize}
\item If $X$ is a D3a reduction that replaces triangle $uvw$ by a new vertex $t$, then let $xt$ and $ty$ be the two edges of $C$ that pass through $t$, and relabel the vertices if necessary so that $ux$ and $vy$ are edges in $G$. Then to form a Hamiltonian cycle in $G$, replace $tx$ and $ty$ in $C$ by the four edges $ux$, $uw$, $vw$, and $vy$. (\autoref{fig:ham}, left.)
\item If $X$ is a D3b reduction of path $uvw$ with apex $x$, and $C$ passes through edge $uw$,
then a Hamiltonian cycle in $G$ may be obtained by replacing $uw$ by $uv$ and $vw$ in $C$. (\autoref{fig:ham}, upper right.)
\item In the remaining case, $X$ is a D3b reduction of path $uvw$ with apex $x$, and $C$ does not pass through edge $uw$. Then (because $u$ and $w$ both have degree three in $H$, and have two incident edges in $C$) the two edges $ux$ and $xw$ must both belong to $C$.
In this case, a Hamiltonian cycle  in $G$ may be obtained by replacing $ux$ by
$uv$ and $vx$. (\autoref{fig:ham}, lower right.)
\end{itemize}
This inductive proof translates directly to an algorithm that reverses the reduction sequence of the  graph and maintains a Hamiltonian cycle for the graph at each step of the reversed reduction sequence. Updating the cycle after each reversed reduction takes constant time so the total time for the algorithm is linear.
\end{proof}

Halin's original motivation for studying Halin graphs was that they provided a natural class of minimally three-connected graphs. This is also true more generally for D3-reducible graphs.

\begin{theorem}
For every D3-reducible graph $G$, and every edge $uv$ of $G$, the graph $G-uv$ formed by deleting $uv$ from $G$ is not 3-vertex-connected.
\end{theorem}

\begin{proof}
Fix a planar embedding of $G$, and let $A$ and $B$ be the faces of $G$ on the two sides of $uv$ in the embedding. We prove more strongly that there exists a Jordan curve that passes through an interior point of $uv$, a vertex $x$ in $A$ (disjoint from $uv$), an a vertex $y$ of $B$ (also disjoint from $uv$) without passing through any other vertices or edges of~$G$. Equivalently, there exists a face $C$, distinct from $A$ and $B$, that includes both $x$ and $y$ as its vertices, so that the desired Jordan curve can be partitioned into three arcs: one in $A$ from $x$ to $uv$, one in $B$ from $uv$ to $y$, and one in $C$ from $y$ back to $x$. Because edge $uv$ crosses this Jordan curve once, it necessarily separates $u$ from $v$, so $x$ and $y$ form a 2-separation of $G-uv$.

If either $u$ or $v$ has degree three, we may take $x$ and $y$ to be the two of its neighbors that are disjoint from $uv$, and $C$ to be the third face (with $A$ and $B$) that is incident to the degree-three vertex. Otherwise, $G$ cannot be $K_4$, so we may assume that it has a D3 reduction $X$ taking it to a smaller D3-reducible graph $H$. Since neither $u$ nor $v$ has degree three, the same edge $uv$ is also present in $H$. By induction, $H$ has faces $A'$, $B'$, and $C'$ and vertices $x'$ and $y'$ with the desired incidence relations to each other.
Whenever $A'$, $B'$, or $C'$ is not the triangular face resulting from a D3b reduction, it has a corresponding face $A$, $B$, or $C$ in $G$. We have the following cases.
\begin{itemize}
\item If $X$ is a D3a reduction whose new supervertex is disjoint from $x'$ and $y'$, then the same vertices $x'$ and $y'$ and (possibly modified) faces $A$, $B$, and $C$ have the same incidence relations in $G$.
\item If $X$ is a D3a reduction that is not disjoint from $x'$ and $y'$, we may assume by symmetry that $x$ is the supervertex formed by contracting a triangle $pqr$. Then in $G$, faces $A$ and $C$ still meet at one of $p$, $q$, or $r$; relabel the triangle if necessary so that they meet at $p$.
Then vertices $p$ and $y$ and faces $A$, $B$, and $C$ have the desired incidence relations.
\item If $X$ is a D3b reduction,  $x'$ and $y'$ are adjacent in $H$, and edge $x'y'$ is not created by reduction $X$, then they remain adjacent in $G$, and either of the two faces incident to them may be chosen as $C$.
\item If $X$ is a D3b reduction,  $x'$ and $y'$ are adjacent in $H$, and edge $x'y'$ is created by reduction $X$ that removed the middle vertex $z$ of a path $x'zy'$, then in $G$ the path $x'zy'$ has the two triangles of the D3b reduction on one side of the path, and a single face $C$ incident to both $x'$ and $y'$ on the other side of the path (since $z$ necessarily has degree three). Again, $x'$, $y'$, $A$, $B$, and $C$ have the desired incidence relation.
\item If $X$ is a D3b-reduction, and $x'$ and $y'$ are not adjacent in $H$, then $C'$ is a face of $H$ with four or more vertices. Then $C'$ corresponds to a face $C$ of $G$ with either the same set of vertices, or with one more vertex (the middle vertex $z$ of the path that was shortened by the D3b reduction). Vertices $x'$, $y'$, and faces $A$, $B$, and $C$ have the desired incidence relation.
\end{itemize}
Thus, in all cases we have shown the existence of three faces and two vertices that, together with edge $uv$, support a Jordan curve separating $u$ from $v$.
\end{proof}

\section{Decomposition, duality, and graph drawing}

\begin{figure}[b]
\centering\includegraphics[height=2in]{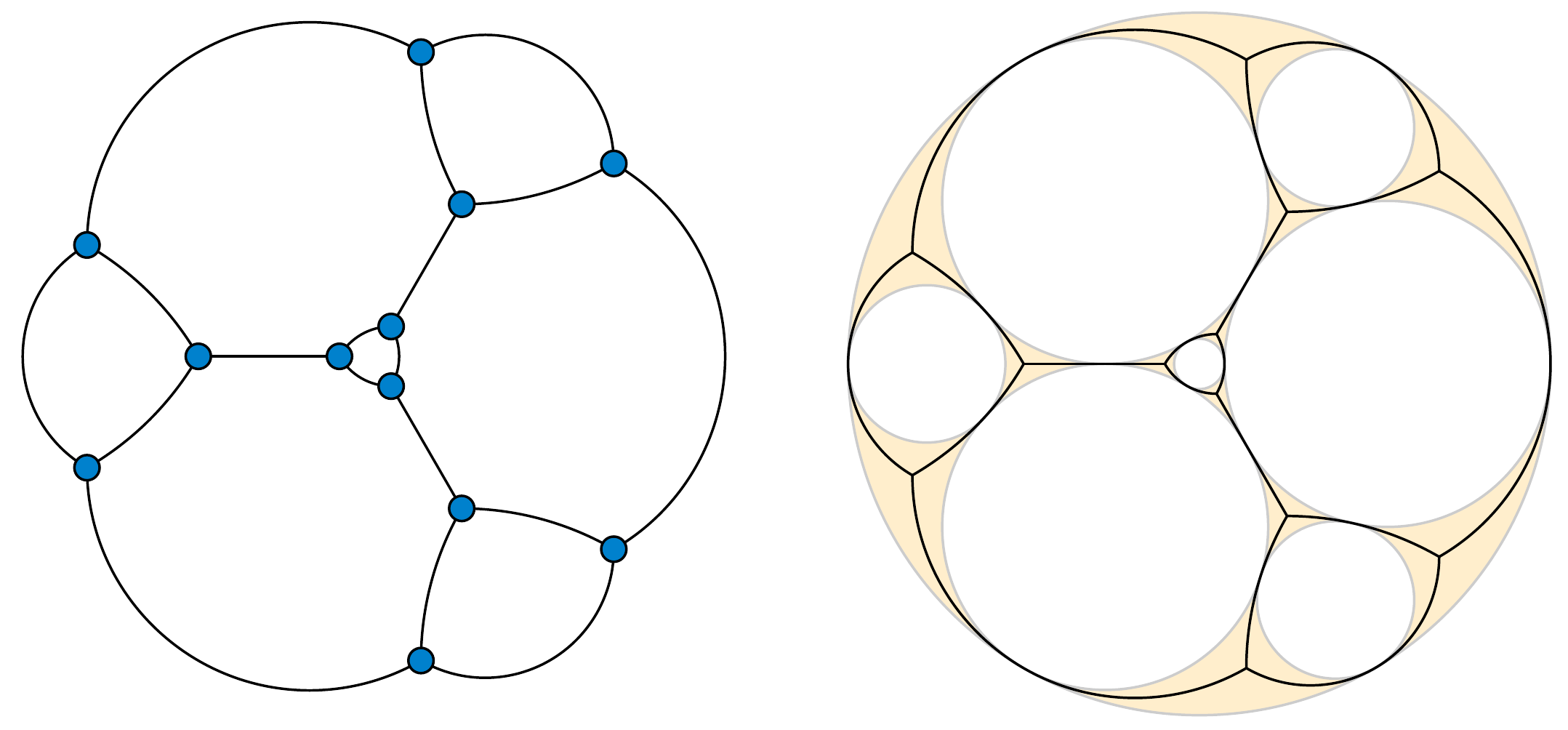}
\caption{A Lombardi drawing of the graph of the truncated tetrahedron (left) and the circle packing used to construct the drawing (right). This graph is D3-reducible and has treewidth~4.}
\label{fig:trunctet}
\end{figure}

Halin graphs all have treewidth three, but this is not true of D3-reducible graphs. In particular, the graph of the truncated tetrahedron (\autoref{fig:trunctet}) is D3-reducible, but has treewidth~four: contracting the six edges that do not belong to triangles produces the octahedral graph $K_{2,2,2}$, which is one of the minor-minimal graphs of treewidth four. However, this example has the largest treewidth possible for these graphs. To prove this, we provide a structural description of the dual graphs of D3-reducible graphs, in terms of clique-sums, operations in which complete subgraphs of pairs of graphs are identified.

Suppose that $G$ and $H$ are two polyhedral graphs in which we have identified an explicit isomorphism between two triangular faces $uvw$ of $G$ and $u'v'w'$ of $H$. Then we may glue these two graphs together by forming the disjoint union of $G$ and $H$ and then collapsing each identified pair of vertices $u$--$u'$, $v$--$v'$, and $w$--$w'$ to a single supervertex. A general clique-sum operation would also allow the removal of some or all of the triangle edges but we do not do this. The result of this gluing operation is a larger polyhedral graph in which the two faces have become a single non-facial triangle. We may perform repeated gluing operations on a collection of polyhedral graphs in the same way, but each triangular face of a graph in the collection may take part only in one of these gluing operations (after which it is no longer a face). We do not allow a graph to be glued to itself, whether it is one of the given graphs or the result of previous gluing steps, because this would not necessarily preserve planarity. Gluing together $p$ polyhedral graphs involves $p-1$ gluing steps (each of which reduces the number of graphs by one), and we can represent these steps abstractly as the edges of a tree whose nodes correspond to the given graphs. The order in which the gluing steps are performed does not affect the result.

\begin{theorem}
\label{thm:glue}
A graph $G$ is the planar dual of a D3-reducible graph if and only if $G$ can be constructed by gluing together a collection of polyhedral graphs, as described above, such that each graph in the collection is a wheel graph (the graph of a tetrahedron or pyramid).
\end{theorem}

\begin{proof}
In one direction, suppose that $G$ is formed by gluing together wheel graphs. We may order the gluing steps so that each step glues a single wheel onto another graph, rather than gluing together two graphs that are themselves the result of other gluing steps. Gluing a four-vertex wheel (the complete graph $K_4$) can be equivalently described as subdividing a triangular face of $G$ into three smaller triangles; the time-reversed operation in the dual graph is a D3a reduction. Gluing a larger wheel may be described as a multiple-step process in which we first glue a four-vertex wheel and then increase the number of vertices in the glued wheel; each of these vertex-increasing operations is the dual to a time-reversed D3b reduction. Thus, reversing and dualizing the sequence of gluing and wheel-increase steps gives us a D3 reduction of the dual graph, showing that it is D3-reducible.

In the other direction, suppose that $G$ is the dual graph of a D3-reducible graph $G'$. As a base case, if $G'$ is $K_4$, $G$ is also $K_4$ and is the graph of a four-vertex wheel. Otherwise, let $X$ be a D3 reduction in $G'$ taking it to a smaller graph $H'$, and let $H$ be the dual of $H'$. By induction, we may assume that $H$ has a representation as a gluing of wheel graphs. If $X$ is a D3a reduction, the dual operation to $X$ un-subdivides a triangle of $G$, and is equivalent to the time-reversal of gluing a four-vertex wheel onto $H$. If $X$ is a D3b reduction of a path $uvw$ with apex $x$ then the vertex $z$ of $H$ dual to triangle $uwx$ has degree three; because each gluing step increases the degree of the glued vertex, this implies that $z$ belongs only to a single wheel of the gluing for $H$. The two vertices $u$ and $w$ are dual to adjacent triangles in $H$, and the D3b reduction is the time-reversed dual of an operation that expands the edge between them into another triangle, increasing the number of vertices of this wheel. Thus, as before, reversing and dualizing the sequence of D3 reductions for $G'$ gives us a sequence of gluing steps for constructing $G$.
\end{proof}

\begin{corollary}
The dual graph of a D3-reducible graph has treewidth three.
\end{corollary}

\begin{proof}
Every wheel graph is a Halin graph, so it has treewidth three, and it is known that clique-sums do not increase the treewidth of the graphs they combine~\cite{Lov-BAMS-06}.
\end{proof}

\begin{corollary}
\label{cor:tw4}
Every D3-reducible graph has treewidth at most four.
\end{corollary}

\begin{proof}
This follows from the fact that the treewidth of a graph is at most one more than the treewidth of its dual graph~\cite{BouMazTod-DM-03}.
\end{proof}

It would be of interest to find a direct proof of \autoref{cor:tw4} that leads to a simple linear-time construction of a width-four tree-decomposition.

The gluing construction for the duals of D3-reducible graphs can also be applied in the construction of graph drawings for the D3-reducible graphs themselves. It is a famous theorem that the vertices of every planar graph can be represented by interior-disjoint disks in such a way that two disks are tangent if and only if the corresponding two vertices are adjacent~\cite{Ste-ICP-05}. For duals of D3-reducible graphs this representation can be chosen with an additional property:

\begin{lemma}
\label{lem:cpack}
Let $G$ be the dual of a D3-reducible graph $G'$. Then its vertices can be represented by interior-disjoint disks as above such that, for every face of $G$, the disks for the vertices of the face are equivalent under a M\"obius transformation to a ring of $d$ congruent disks with cocircular centers.
\end{lemma}

\begin{proof}
Every wheel graph has a representation of this form. If two polyhedral graphs $G$ and $H$ are glued together on triangular faces, their disk representations may also be obtained by gluing together the representations for $G$ and $H$, applying a M\"obius transformation to the representation of $G$ make the three disks for the gluing face have the same size and position as they do in the representation of $H$.  The result follows from \autoref{thm:glue}.
\end{proof}

\begin{corollary}
Every D3-reducible graph has a planar \emph{Lombardi drawing}, a drawing in which the vertices are represented by points and the edges are represented by circular arcs that meet at equal angles at each vertex.
\end{corollary}

\begin{proof}
A construction of the author for Lombardi drawings of cubic graphs~\cite{Epp-DCG-14} forms a disk representation for the dual graph.
It defines a distance function from points to these disks, where the distance from point $p$ to disk $D$ is the radius of two congruent disks that are tangent to each other at $p$ and also both tangent to $D$, and constructs the minimization diagram of this distance, a partition of the plane into cells within which one of the disks is closer than all the others~\cite[Sec.~3]{Epp-DCG-14}. This minimization diagram has circular arcs for boundaries~\cite[Lem.~2]{Epp-DCG-14}, which in the case of cubic graphs necessarily meet at angles of $2\pi/3$, forming a Lombardi drawing of the original graph. It is invariant under M\"obius transformations of the plane: transforming a circle packing and then constructing the minimization diagram, or constructing the diagram first and then transforming it, produces the same result~\cite[Lem.~1]{Epp-DCG-14}.

For a D3-reducible graph we use the same minimization diagram for the circle packing of \autoref{lem:cpack}.
The resulting minimization diagram again has piecewise-circular boundaries between cells and is invariant under M\"obius transformations of the plane. By symmetry and M\"obius invariance, these boundaries must meet at equal angles at a point within each face of the dual graph, forming a Lombardi drawing of the primal D3-reducible graph.
\end{proof}

An example of a Lombardi drawing constructed in this way for the graph of the truncated tetrahedron is shown in \autoref{fig:trunctet}. This graph has all vertices of degree three, a property already known to guarantee the existence of a Lombardi drawing~\cite{Epp-DCG-14}, but the same method works as well for D3-reducible graphs with vertices of higher degree. However, it is not true that the duals of D3-reducible graphs always have planar Lombardi drawings; indeed, it is known that some planar 3-trees (a special case of the duals of D3-reducible graphs) do not have such drawings~\cite{DunEppGoo-GD-11}.

\section{Halin graph recognition}

We return to our motivating problem of developing a simple algorithm for Halin graph recognition.
We have already seen that the D3 reductions will apply to any Halin graph; however, they also apply to some graphs that are not Halin graphs, so we need to modify the reduction process to avoid
these more general reductions. The key observation is the following:

\begin{observation}
\label{obs:reduction-is-at-leaves}
Let $G$ be a Halin graph, constructed from a tree $T$ with outer cycle $C$. Then every D3a reduction in $G$ must form a simpler Halin graph by removing the two children from a node of $T$ that has only two children, both leaves, and every D3b reduction must form a simpler graph by removing a middle leaf child from a tree node that has three consecutive leaves among its children.
\end{observation}

A Halin graph may have more than one decomposition into a tree and a cycle, but if so this observation applies simultaneously to all of these decompositions. The reason the observation holds is that both the D3a and D3b reductions require the presence of a triangle in $G$, and the only triangles in a Halin graph can be the ones formed by two leaf edges of $T$ and an edge of $C$.

Intuitively, whenever we perform a reduction in a Halin graph, \autoref{obs:reduction-is-at-leaves} gives us more information about the set of vertices that belong to the outer cycle. We will use this information to check whether an arbitrary D3-reducible graph is Halin.
Our algorithm for testing whether a graph $G$ is Halin follows the same outline as the algorithm for testing whether $G$ is D3-reducible, with the following modifications:
\begin{itemize}
\item We maintain a set $K$ of vertices that are known to belong to the outer cycle of a Halin graph representation of $G$; initially, $K$ is empty.
\item Whenever we perform a D3a reduction of a triangle $uvw$, replacing it by a vertex $x$, we first check whether all three of $u$, $v$, and $w$ belong to $K$; if they do, we forbid this reduction.
Otherwise we perform the reduction and then add $x$ to $K$. Additionally, if any one or two of $u$, $v$, or $w$ were already members of $K$, we add the neighbor or neighbors of these known-outer vertices to $K$ as well. Examples of this reduction are shown by the two rightmost arrows in \autoref{fig:nonhalin}.
\item Whenever we perform a D3b reduction of a path $uvw$ with apex $x$, removing the middle vertex $v$, we check whether $x$ is in $K$; if so, we forbid the reduction. Otherwise, we add $u$ and $w$ to $K$. Additionally, if the reduction reduces the graph to $K_4$, we add the fourth vertex (the one that is not $x$) to $K$. Examples of this reduction are shown by the left and upper middle arrows in \autoref{fig:nonhalin}. The lower right graph of the figure gives an example in which there are four potential D3b reductions, but all are forbidden because their shared apex is in~$K$.
\item When an irreducible graph is reached, as well as checking that it is isomorphic to $K_4$, we check that it has at least one vertex that does not belong to $K$. If so, we recognize it as a Halin graph; otherwise we do not. The upper right graph of \autoref{fig:nonhalin} gives an example in which the recognition algorithm can reach a $K_4$ graph but fails to recognize the graph as Halin because all vertices belong to~$K$.
\end{itemize}

\begin{figure}[t]
\centering
\includegraphics[scale=0.5]{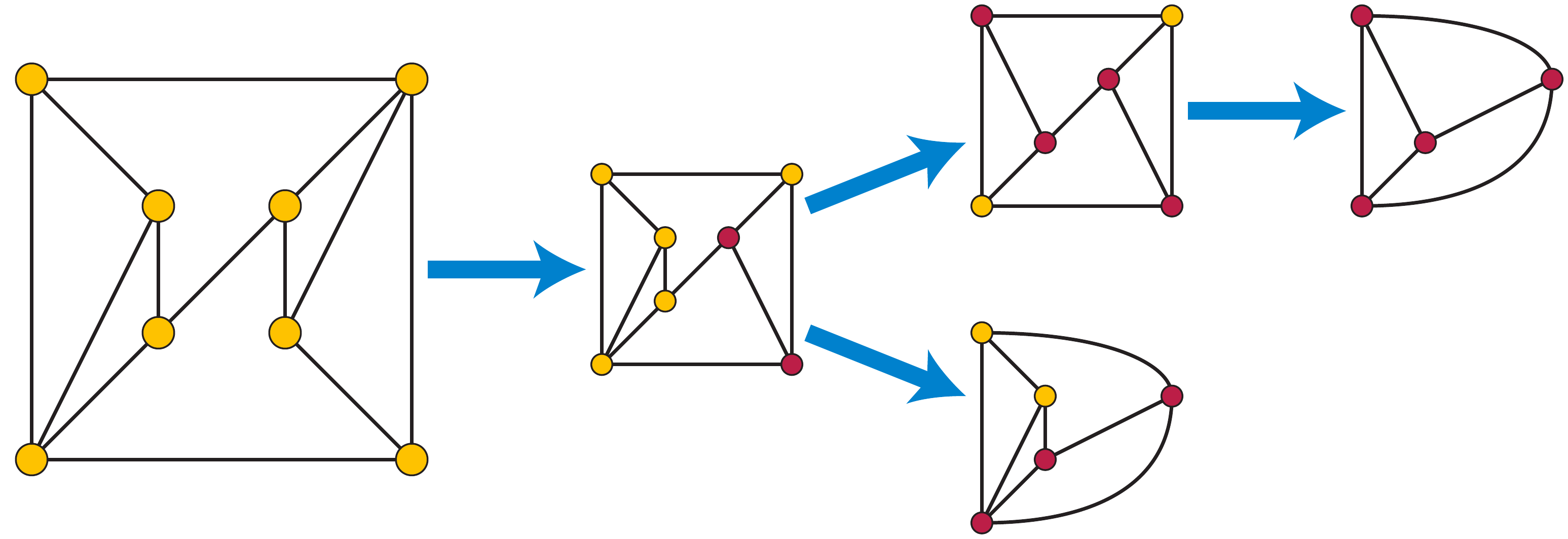}
\caption{A D3-reducible graph that is not Halin and its possible reduction sequences (up to isomorphism) under the Halin graph recognition algorithm. The set of known-outer vertices is shown in red.}
\label{fig:nonhalin}
\end{figure}

\begin{lemma}
\label{lem:Halin-subset}
If $G$ is a Halin graph then the
modified algorithm described above will never forbid a reduction. The graph remaining after each step will also be Halin, and the intersection of that graph with $K$ will consist only of vertices that belong to the outer cycle of every decomposition of this graph into a tree and a cycle.
\end{lemma}

\begin{proof}
We prove the lemma by induction on the number of steps of the algorithm; initially $K$ is empty and the result holds vacuously.

At each D3a reduction, in which a triangle $uvw$ is contracted, one of the triangle vertices (say $u$) must be the parent of the other two vertices $v$ and $w$, which must be leaves in $T$. By the induction hypothesis, $u$ is not in $K$ prior to the D3a reduction so the reduction will not be forbidden. After the reduction, the removal of two leaves causes the contracted supervertex to become a leaf in the reduced version of $T$, so adding it to $K$ is valid. And the only edges in $T$ incident to $v$ and $w$ are the ones connecting them to their parent $u$, so if this step causes the neighbors of $v$ and $w$ to become added to $K$ the result is again valid.

At each D3b reduction, in which a path $uvw$ with apex $x$ is shortened, the three vertices $u$, $v$, and $w$ must all be children of $x$ in $T$. By the induction hypotheses $x$ will not belong to $K$ and the reduction will not be forbidden. Vertices $u$ and $w$ remain leaf children of $x$ after the reduction, so adding $u$ and $w$ to $K$ is valid.
\end{proof}

\begin{corollary}
\label{cor:no-false-neg}
If $G$ is a Halin graph then the algorithm described above will correctly recognize $G$ as being a Halin graph.
\end{corollary}

\begin{proof}
By Lemma~\ref{lem:Halin-subset}, $G$ will be reduced to an irreducible graph, which must be $K_4$, and this graph must have a face that forms a superset of $K$.
Therefore, there will be at least one vertex of the irreducible graph that is not in $K$, so the termination condition of the algorithm is met and the algorithm will necessarily recognize $G$ as Halin.
\end{proof}

\begin{lemma}
\label{lem:no-false-pos}
If $G$ is recognized by the algorithm described above, then it is indeed a Halin graph, and has a decomposition into a tree $T$ and a cycle $C$ in which the vertices of $K$ all belong to the cycle. \end{lemma}

\begin{proof}
We prove the result by induction on the size of $G$. If $G$ is irreducible, it can only be recognized if it is $K_4$, which is indeed a Halin graph. Otherwise, suppose that $X$ is the first reduction found by the algorithm, and let $H$ be the smaller graph formed from $G$ by reduction $X$. By induction, $H$ is Halin, with a decomposition into a tree $T'$ and cycle $C'$ with $K\cap H\subset C'$. We have the following cases:
\begin{itemize}
\item If $X$ is a D3a reduction of triangle $uvw$, replacing these three vertices by a single vertex $x$, then (because the algorithm adds $x$ to $K$) $x$ must belong to $C'$, and must form a leaf of the tree $T'$.  Two edges $xu'$ and $xv'$ must belong to $C'$, and the third edge $xw'$ cannot (because a cycle has degree two at each vertex). To form a Halin graph decomposition of $G$, we replace $x$ in $T'$ by $w$, and add $u$ and $v$ as children of $w$ to form the tree $T$. We form the cycle $C$ by replacing the edges $xu'$ and $xv'$ in $C'$ by the three edges $u'u$, $uv$, and $vv'$. The resulting tree and cycle decompose $G$ in the manner required of a Halin graph, so $G$ is Halin.

Vertices $u$ and $v$ belong to the cycle $C$ but vertex $w$ does not. Vertex $w$ cannot have been part of $K$ prior to performing reduction $X$,
because if it were then in $H$ vertices $x$ and $w'$ would both belong to $K$, forcing edge $xw'$ to belong to $C'$ (because $H$ is Halin and in a Halin graph every edge between leaf vertices belongs to the outer cycle) contradictory to our assumption. Therefore, in $G$ it remains true that the vertices of $K$ all belong to cycle $C$.
\item If $X$ is a D3b reduction of path $uvw$ with apex $x$, removing $v$ and shortening the path,
then after the reduction $u$ and $w$ belong to $K$, so they must both be leaf vertices of $T'$. The edge $uw$ connecting them must belong to $C'$. The other two edges $ux$ and $wx$ of the triangle $uwx$ in $H$ must be leaf edges of $T'$, for the only other possibility (that together with $uw$ they form the outer cycle of a $K_4$ graph) is prevented by the special handling of a $K_4$ in the D3b reduction. We form the cycle $C$ by replacing edge $uw$ by the path $uvw$, and we form the cycle $T$ by adding $v$ as a leaf child of~$x$. The resulting tree and cycle decompose $G$ in the manner required of a Halin graph, so $G$ is Halin.

In this case, all vertices other than $v$ either belong to both $C$ and $C'$, or belong to neither.
Therefore, the condition that $K$ is a subset of $C$ follows from the induction hypothesis that $K\cap H$ is a subset of $C'$, together with the fact that the construction places $v$~in~$C$.\qedhere
\end{itemize}
\end{proof}

\begin{theorem}
The algorithm described above correctly recognizes Halin graphs in linear time,
and can be modified to construct a decomposition of a Halin graph into a tree and a cycle in linear time.
\end{theorem}

\begin{proof}
The correctness of the algorithm follows from \autoref{cor:no-false-neg} and \autoref{lem:no-false-pos}. The modifications to the D3-reducibility algorithm add constant time per reduction so the time analysis is the same as for testing D3-reducibility. To construct a decomposition, we reverse the steps of the reduction and use \autoref{lem:no-false-pos} to maintain at each step of the reversed sequence a decomposition of the Halin graph from that step of the sequence; again, this adds constant time per step of the reduction.
\end{proof}

\section{Implementation}

To support our claim that the reduction-based method described here leads to simple and implementable algorithms, we developed a proof-of-concept implementation of our algorithms in the Python programming language, including the algorithms for testing D3-reducibility, finding Hamiltonian cycles in D3-reducible graphs, testing whether a graph is Halin, and finding the set of leaf nodes of an (arbitrarily chosen) decomposition of a Halin graph into a tree and a cycle.

\subsection{Graph representation}

To support constant-time graph reduction operations, adjacency tests, neighbor listing operations, and neighbor counting operations, we use a modified version of a graph representation scheme suggested by van Rossum~\cite{Ros-PP-98}. In van Rossum's representation, a graph is a Python dictionary object (a hash table) with vertices as its keys and with Python lists (dynamic arrays) of neighboring vertices as the associated values. There is no need for special vertex objects: vertices in this representation are allowed to be any type of object that can be used as keys in a dictionary, such as integers or strings.

However, a Python list does not allow constant-time membership testing, nor the constant-time removal of a vertex from a list of neighbors without knowing its position in the list. Therefore, we modify van Rossum's representation by
representing a graph as a dictionary with the vertices as keys and with Python sets of neighboring vertices as the associated values. The set data type was introduced to Python subsequently to van Rossum's original proposal for this representation. Using sets in place of lists allows more flexible and fast addition, removal, and membership testing in each vertex neighborhood.

\subsection{Software architecture}

In order to maximize the code re-use of our implementation, we designed it to have a central core that finds and performs D3 reductions on a graph, and that takes as arguments callback routines that either modify the sequence of reductions that can be performed or record information about the reductions as they are performed.

More specifically, the main subroutine of our implementation takes three arguments: a list of \emph{triangle hooks}, another list of \emph{path hooks}, and a \emph{finalizer}. These arguments have the following meanings:

\begin{itemize}
\item The triangle hooks are a list of subroutines that are called, in the order given by the list, before performing any D3a reduction. These subroutines take the graph and seven vertices as arguments: the six vertices forming the configuration to be reduced, and a seventh vertex that will replace the central triangle in this configuration. They return a Boolean value, true if the reduction should be allowed to happen and false otherwise. If any triangle hook returns false, the remaining ones on the list are not called; otherwise, all are called. As well as being used to constrain which D3a reductions occur, these hooks may also be used to record information about the sequence of reductions performed by the algorithm.
\item The path hooks are another list of subroutines, used in the same way for D3b reductions as the triangle hooks are used for D3a reductions. They each take five arguments: the graph, three path vertices and apex of a D3b reduction.
\item The finalizer takes as input the irreducible graph after all reductions are complete, and produces as output the value that should be returned as the result of the overall computation.
\end{itemize}

\noindent
We  implemented several additional sets of subroutines to be used as these arguments:
\begin{itemize}
\item To recognize Halin graphs, we use triangle and path hooks that maintain a set of vertices required to be part of the outer cycle, and that prevent reductions inconsistent with this requirement.
We use a finalizer that returns true when the graph can be reduced to $K_4$ without requiring all four vertices to be part of the outer cycle, and false otherwise.
\item We also implemented alternative recognition algorithms based on D3 reductions for testing whether a given graph is the dual graph of a planar 3-tree, or whether it is a wheel. These algorithms use trivial triangle or path hooks that prevent any D3b reduction in the case of dual 3-trees, or that prevent any D3a reduction in the case of wheels, together with a finalizer that merely checks whether the reduced graph is $K_4$.
\item We implemented a pair of triangle and path hooks that record a sequence of D3 reductions. We use these hooks as part of a subroutine that recursively calls any D3-reduction based recognition algorithm (such as our Halin graph recognition algorithm), reverses the recorded sequence of reductions made during the algorithm, and then calls a given pair of triangle and path functions (with the same arguments as the triangle and path hooks) in the order given by the reversed sequence. This can be used to inductively construct structures associated with Halin or D3-reducible graphs.
\item We implemented a method for finding the set of leaf vertices in a decomposition of a Halin graph into a cycle and a tree, using our subroutine for inductive construction together with additional triangle and path subroutines that update this set of leaf vertices through the reversal of any D3 reduction. For a Halin graph that has more than one valid decomposition, one is chosen arbitrarily.
\item We also implemented another method for constructing a Hamiltonian cycle in a D3-reducible graph, again using our subroutine for inductive construction together with additional triangle and path subroutines that update the Hamiltonian cycle through the reversal of any D3 reduction.
\end{itemize}

\subsection{Code size and testing}

Open-source Python code for our implementation is available online at
\url{http://www.ics.uci.edu/~eppstein/PADS/Halin.py}.

In our implementation, not counting comments, whitespace, and sanity checks, the basic D3 reducibility test takes 65 lines of code, and the subroutines to record and reverse a sequence of reductions take 12 lines of code. The additional subroutines for Halin graph recognition take 26 lines of code, the subroutines for finding the leaf vertices of a Halin graph take 15 lines of code, and the subroutines for constructing a Hamiltonian cycle take 28 lines of code. We believe that this code size is substantially smaller than what would be required for a Halin graph recognition algorithm based on general linear-time planarity testing methods.

We checked the correctness of our implementations by unit tests that run them and compare their output with the known correct output for several graphs. Our test cases include examples of Halin graphs, D3-reducible but non-Halin graphs, and non-D3-reducible graphs, on up to 40 vertices. The size of these test graphs was limited by the need to have independent human verification of the correctness of the results rather than by the performance of the algorithms.

Because the implementation of Python that we used is a slow interpreted language, we did not attempt to measure the runtime of our algorithms, as we feel that this measurement would not provide useful information about the efficiency of the same algorithms when implemented in a higher-performance environment.

\section{Conclusions and open problems}

We have developed simple and implementable algorithms for recognizing Halin graphs and for several related problems. These algorithms led us to the definition of a class of graphs, the D3-reducible graphs, that generalize the Halin graphs and share many of their important properties.

It would be of interest to determine more precisely where the D3-reducible graphs fit within the complicated hierarchy of known graph classes. For instance, as well as being a subclass of the polyhedral graphs (which also include the D3-reducible graphs) and the planar partial 3-trees (which don't), the Halin graphs are a subclass of the intersection graphs of rectangles~\cite{ChaFraSur-DM-09}. Is this also true of the  D3-reducible graphs?

\bibliographystyle{abuser}
\bibliography{halin}

\begin{thebibliography}{10}
\urlstyle{rm}

\bibitem{ArnCouPro-JACM-93}
S.~Arnborg, B.~Courcelle, A.~Proskurowski, and D.~Seese.
\newblock {An algebraic theory of graph reduction}.
\newblock {\em Journal of the ACM} 40(5):1134{--}1164, 1993,
  \href{http://dx.doi.org/10.1145/174147.169807}%
{doi:\nolinkurl{10.1145/174147.169807}},
  \href{https://www.ams.org/mathscinet-getitem?mr=1368961}%
{MR1368961}.

\bibitem{ArnPro-SJADM-86}
S.~Arnborg and A.~Proskurowski.
\newblock {Characterization and recognition of partial $3$-trees}.
\newblock {\em SIAM Journal on Algebraic and Discrete Methods} 7(2):305{--}314,
  1986, \href{http://dx.doi.org/10.1137/0607033}%
{doi:\nolinkurl{10.1137/0607033}},
  \href{https://www.ams.org/mathscinet-getitem?mr=830649}%
{MR830649}.

\bibitem{BarGru-MFGT-69}
D.~W. Barnette and B.~Gr{\"u}nbaum.
\newblock {On Steinitz's theorem concerning convex $3$-polytopes and on some
  properties of planar graphs}.
\newblock {\em The Many Facets of Graph Theory (Proc. Conf., Western Mich.
  Univ., Kalamazoo, Mich., 1968)}, pp.~27{--}40. Springer-Verlag, Lecture Notes
  in Mathematics 110, 1969, \href{http://dx.doi.org/10.1007/BFb0060102}%
{doi:\nolinkurl{10.1007/BFb0060102}},
  \href{https://www.ams.org/mathscinet-getitem?mr=0250916}%
{MR0250916}.

\bibitem{Bod-ICALP-88}
H.~L. Bodlaender.
\newblock {Dynamic programming on graphs with bounded treewidth}.
\newblock {\em Proceedings of the 15th International Colloquium on Automata,
  Languages and Programming}, pp.~105{--}118. Springer-Verlag, Lecture Notes in
  Computer Science 317, 1988,
  \href{http://dx.doi.org/10.1007/3-540-19488-6_110}%
{doi:\nolinkurl{10.1007/3-540-19488-6_110}}.

\bibitem{Bod-SJC-96}
H.~L. Bodlaender.
\newblock {A linear time algorithm for finding tree-decompositions of small
  treewidth}.
\newblock {\em SIAM Journal on Computing} 25(6):1305{--}1317, 1996,
  \href{http://dx.doi.org/10.1137/S0097539793251219}%
{doi:\nolinkurl{10.1137/S0097539793251219}}.

\bibitem{BodThi-Algs-99}
H.~L. Bodlaender and D.~M. Thilikos.
\newblock {Graphs with branchwidth at most three}.
\newblock {\em Journal of Algorithms} 32(2):167{--}194, 1999,
  \href{http://dx.doi.org/10.1006/jagm.1999.1011}%
{doi:\nolinkurl{10.1006/jagm.1999.1011}},
  \href{https://www.ams.org/mathscinet-getitem?mr=1698980}%
{MR1698980}.

\bibitem{BooLue-JCSS-76}
K.~S. Booth and G.~S. Lueker.
\newblock {Testing for the consecutive ones property, interval graphs, and
  graph planarity using $PQ$-tree algorithms}.
\newblock {\em Journal of Computer and System Sciences} 13(3):335{--}379, 1976,
  \href{https://www.ams.org/mathscinet-getitem?mr=0433962}%
{MR0433962}.

\bibitem{BouMazTod-DM-03}
V.~Bouchitt{\'e}, F.~Mazoit, and I.~Todinca.
\newblock {Chordal embeddings of planar graphs}.
\newblock {\em Discrete Mathematics} 273(1-3):85{--}102, 2003,
  \href{http://dx.doi.org/10.1016/S0012-365X(03)00230-9}%
{doi:\nolinkurl{10.1016/S0012-365X(03)00230-9}},
  \href{https://www.ams.org/mathscinet-getitem?mr=2025943}%
{MR2025943}.

\bibitem{BoyMyr-JGAA-04}
J.~M. Boyer and W.~J. Myrvold.
\newblock {On the cutting edge: simplified $O(n)$ planarity by edge addition}.
\newblock {\em Journal of Graph Algorithms and Applications} 8(3):241{--}273,
  2004, \href{http://dx.doi.org/10.7155/jgaa.00091}%
{doi:\nolinkurl{10.7155/jgaa.00091}}.

\bibitem{Bru-JCTB-04}
H.~Bruhn.
\newblock {The cycle space of a 3-connected locally finite graph is generated
  by its finite and infinite peripheral circuits}.
\newblock {\em Journal of Combinatorial Theory (Series B)} 92(2):235{--}256,
  2004, \href{http://dx.doi.org/10.1016/j.jctb.2004.03.005}%
{doi:\nolinkurl{10.1016/j.jctb.2004.03.005}},
  \href{https://www.ams.org/mathscinet-getitem?mr=2099143}%
{MR2099143}.

\bibitem{ChaFraSur-DM-09}
L.~S. Chandran, M.~C. Francis, and S.~Suresh.
\newblock {Boxicity of Halin graphs}.
\newblock {\em Discrete Mathematics} 309(10):3233{--}3237, 2009,
  \href{http://dx.doi.org/10.1016/j.disc.2008.09.037}%
{doi:\nolinkurl{10.1016/j.disc.2008.09.037}},
  \href{https://www.ams.org/mathscinet-getitem?mr=2526741}%
{MR2526741}.

\bibitem{ChiNisAbe-JCSS-85}
N.~Chiba, T.~Nishizeki, A.~Abe, and T.~Ozawa.
\newblock {A linear algorithm for embedding planar graphs using $PQ${--}trees}.
\newblock {\em Journal of Computer and System Sciences} 30(1):54{--}76, 1985,
  \href{http://dx.doi.org/10.1016/0022-0000(85)90004-2}%
{doi:\nolinkurl{10.1016/0022-0000(85)90004-2}}.

\bibitem{Cou-IC-90}
B.~Courcelle.
\newblock {The monadic second-order logic of graphs. I. Recognizable sets of
  finite graphs}.
\newblock {\em Information and Computation} 85(1):12{--}75, 1990,
  \href{http://dx.doi.org/10.1016/0890-5401(90)90043-H}%
{doi:\nolinkurl{10.1016/0890-5401(90)90043-H}},
  \href{https://www.ams.org/mathscinet-getitem?mr=1042649}%
{MR1042649}.

\bibitem{Duf-JMAA-65}
R.~J. Duffin.
\newblock {Topology of series-parallel networks}.
\newblock {\em Journal of Mathematical Analysis and Applications}
  10:303{--}318, 1965,
  \href{https://www.ams.org/mathscinet-getitem?mr=0175809}%
{MR0175809}.

\bibitem{DunEppGoo-JGAA-12}
C.~A. Duncan, D.~Eppstein, M.~T. Goodrich, S.~G. Kobourov, and
  M.~N{\"o}llenburg.
\newblock {Lombardi drawings of graphs}.
\newblock {\em Journal of Graph Algorithms and Applications} 16(1):85{--}108,
  2012, \href{http://dx.doi.org/10.7155/jgaa.00251}%
{doi:\nolinkurl{10.7155/jgaa.00251}},
  \href{https://www.ams.org/mathscinet-getitem?mr=2872431}%
{MR2872431}.

\bibitem{DunEppGoo-GD-11}
C.~A. Duncan, D.~Eppstein, M.~T. Goodrich, S.~G. Kobourov, and
  M.~N{\"o}llenburg.
\newblock {Planar and poly-arc Lombardi drawings}.
\newblock {\em Proc. 19th Int. Symp. Graph Drawing (GD 2011)}, pp.~308{--}319.
  Springer-Verlag, Lecture Notes in Computer Science 7034, 2012,
  \href{http://dx.doi.org/10.1007/978-3-642-25878-7_30}%
{doi:\nolinkurl{10.1007/978-3-642-25878-7_30}},
  \href{https://www.ams.org/mathscinet-getitem?mr=2928294}%
{MR2928294}.

\bibitem{ElbJakTan-FOCS-10}
M.~Elberfeld, A.~Jakoby, and T.~Tantau.
\newblock {Logspace versions of the theorems of Bodlaender and Courcelle}.
\newblock {\em Proc. 51st Symp. Foundations of Computer Science (FOCS 2010)},
  pp.~143{--}152, 2010, \href{http://dx.doi.org/10.1109/FOCS.2010.21}%
{doi:\nolinkurl{10.1109/FOCS.2010.21}}.

\bibitem{Epp-DCG-14}
D.~Eppstein.
\newblock {A M{\"o}bius-invariant power diagram and its applications to soap
  bubbles and planar Lombardi drawing}.
\newblock {\em Discrete {\&} Computational Geometry} 52(3):515{--}550, 2014,
  \href{http://dx.doi.org/10.1007/s00454-014-9627-0}%
{doi:\nolinkurl{10.1007/s00454-014-9627-0}},
  \href{https://www.ams.org/mathscinet-getitem?mr=3257673}%
{MR3257673}.

\bibitem{Flu-PhD-97}
B.~de~Fluiter.
\newblock {\em {Algorithms for Graphs of Small Treewidth}}.
\newblock Ph.D. thesis, Utrecht University, 1997.

\bibitem{FomThi-JDA-06}
F.~V. Fomin and D.~M. Thilikos.
\newblock {A 3-approximation for the pathwidth of Halin graphs}.
\newblock {\em Journal of Discrete Algorithms} 4(4):499{--}510, 2006,
  \href{http://dx.doi.org/10.1016/j.jda.2005.06.004}%
{doi:\nolinkurl{10.1016/j.jda.2005.06.004}},
  \href{https://www.ams.org/mathscinet-getitem?mr=2577677}%
{MR2577677}.

\bibitem{FraOssRos-IJFCS-06}
H.~de~Fraysseix, P.~Ossona~de Mendez, and P.~Rosenstiehl.
\newblock {Tr{\'e}maux trees and planarity}.
\newblock {\em International Journal of Foundations of Computer Science}
  17(5):1017{--}1030, 2006, \href{http://dx.doi.org/10.1142/S0129054106004248}%
{doi:\nolinkurl{10.1142/S0129054106004248}}.

\bibitem{GutMut-GD-00}
C.~Gutwenger and P.~Mutzel.
\newblock {A linear time implementation of SPQR-trees}.
\newblock {\em Proc. 8th International Symposium on Graph Drawing (GD 2000)},
  pp.~77{--}90. Springer-Verlag, Lecture Notes in Computer Science 1984, 2001,
  \href{http://dx.doi.org/10.1007/3-540-44541-2_8}%
{doi:\nolinkurl{10.1007/3-540-44541-2_8}}.

\bibitem{Hal-CMA-71}
R.~Halin.
\newblock {Studies on minimally $n$-connected graphs}.
\newblock {\em Combinatorial Mathematics and its Applications (Proc. Conf.,
  Oxford, 1969)}, pp.~129{--}136. Academic Press, 1971,
  \href{https://www.ams.org/mathscinet-getitem?mr=0278980}%
{MR0278980}.

\bibitem{HopTar-SJC-73}
J.~Hopcroft and R.~E. Tarjan.
\newblock {Dividing a graph into triconnected components}.
\newblock {\em SIAM Journal on Computing} 2(3):135{--}158, 1973,
  \href{http://dx.doi.org/10.1137/0202012}%
{doi:\nolinkurl{10.1137/0202012}}.

\bibitem{HopTar-JACM-74}
J.~Hopcroft and R.~E. Tarjan.
\newblock {Efficient planarity testing}.
\newblock {\em Journal of the ACM} 21(4):549{--}568, 1974,
  \href{http://dx.doi.org/10.1145/321850.321852}%
{doi:\nolinkurl{10.1145/321850.321852}}.

\bibitem{LofNol-GD-12}
M.~L{\"o}ffler and M.~N{\"o}llenburg.
\newblock {Planar Lombardi drawings of outerpaths}.
\newblock {\em Proc. 20th Int. Symp. Graph Drawing (GD 2012)}, pp.~561{--}562.
  Springer-Verlag, Lecture Notes in Computer Science 7704, 2013,
  \href{http://dx.doi.org/10.1007/978-3-642-36763-2_53}%
{doi:\nolinkurl{10.1007/978-3-642-36763-2_53}}.

\bibitem{Lov-BAMS-06}
L.~Lov{\'a}sz.
\newblock {Graph minor theory}.
\newblock {\em Bulletin of the American Mathematical Society} 43(1):75{--}86,
  2006, \href{http://dx.doi.org/10.1090/S0273-0979-05-01088-8}%
{doi:\nolinkurl{10.1090/S0273-0979-05-01088-8}},
  \href{https://www.ams.org/mathscinet-getitem?mr=2188176}%
{MR2188176}.

\bibitem{Ros-PP-98}
G.~van Rossum.
\newblock {Python Patterns {---} Implementing Graphs}.
\newblock Python essays, Python Software Foundation, 1998--2003,
  \url{https://www.python.org/doc/essays/graphs/}.

\bibitem{Sch-Algo-12}
J.~M. Schmidt.
\newblock {Construction sequences and certifying 3-connectivity}.
\newblock {\em Algorithmica} 62(1-2):192{--}208, 2012,
  \href{http://dx.doi.org/10.1007/s00453-010-9450-9}%
{doi:\nolinkurl{10.1007/s00453-010-9450-9}},
  \href{https://www.ams.org/mathscinet-getitem?mr=2886041}%
{MR2886041}.

\bibitem{Sch-MFCS-13}
J.~M. Schmidt.
\newblock {A planarity test via construction sequences}.
\newblock {\em Proc. 38th Int. Symp. Mathematical Foundations of Computer
  Science (MFCS 2013)}, pp.~765{--}776. Springer-Verlag, Lecture Notes in
  Computer Science 8087, 2013,
  \href{http://dx.doi.org/10.1007/978-3-642-40313-2_67}%
{doi:\nolinkurl{10.1007/978-3-642-40313-2_67}},
  \href{https://www.ams.org/mathscinet-getitem?mr=3126255}%
{MR3126255}.

\bibitem{ShiHsu-TCS-99}
W.~K. Shih and W.~L. Hsu.
\newblock {A new planarity test}.
\newblock {\em Theoretical Computer Science} 223(1{--}2):179{--}191, 1999,
  \href{http://dx.doi.org/10.1016/S0304-3975(98)00120-0}%
{doi:\nolinkurl{10.1016/S0304-3975(98)00120-0}}.

\bibitem{SkoSys-ZM-87}
M.~Skowro{\'n}ska and M.~M. Sys{\l}o.
\newblock {Hamiltonian cycles in skirted trees}.
\newblock {\em Proceedings of the International Conference on Combinatorial
  Analysis and its Applications (Pokrzywna, 1985)}, 3-4 edition,
  pp.~599{--}610. Polska Akademia Nauk. Instytut Matematyczny, Zastosowania
  Matematyki Applicationes Mathematicae~19, 1987,
  \href{https://www.ams.org/mathscinet-getitem?mr=951375}%
{MR951375}.

\bibitem{Ste-EMW-22}
E.~Steinitz.
\newblock {Polyeder und Raumeinteilungen}.
\newblock {\em Encyclop{\"a}die der mathematischen Wissenschaften, Band 3
  (Geometries)}, pp.~1{--}139, 1922.

\bibitem{Ste-ICP-05}
K.~Stephenson.
\newblock {\em {Introduction to circle packing, the theory of discrete analytic
  functions}}.
\newblock Cambridge University Press, Cambridge, 2005.

\bibitem{SysPro-GT-83}
M.~M. Sys{\l}o and A.~Proskurowski.
\newblock {On Halin graphs}.
\newblock {\em Graph Theory: Proceedings of a Conference held in Lag{\'o}w,
  Poland, February 10{--}13, 1981}, pp.~248{--}256. Springer-Verlag, Lecture
  Notes in Mathematics 1018, 1983, \href{http://dx.doi.org/10.1007/BFb0071635}%
{doi:\nolinkurl{10.1007/BFb0071635}}.

\bibitem{ValTarLaw-SJC-82}
J.~Valdes, R.~E. Tarjan, and E.~L. Lawler.
\newblock {The recognition of series parallel digraphs}.
\newblock {\em SIAM Journal on Computing} 11(2):298{--}313, 1982,
  \href{http://dx.doi.org/10.1137/0211023}%
{doi:\nolinkurl{10.1137/0211023}},
  \href{https://www.ams.org/mathscinet-getitem?mr=652904}%
{MR652904}.

\bibitem{Win-DAM-87}
P.~Winter.
\newblock {Steiner problem in Halin networks}.
\newblock {\em Discrete Applied Mathematics} 17(3):281{--}294, 1987,
  \href{http://dx.doi.org/10.1016/0166-218X(87)90031-X}%
{doi:\nolinkurl{10.1016/0166-218X(87)90031-X}},
  \href{https://www.ams.org/mathscinet-getitem?mr=890638}%
{MR890638}.

\end{thebibliography}

\end{document}